\documentclass[11pt]{article}
\usepackage{times}
\usepackage{fullpage}
\usepackage{epsfig}
\usepackage{latexsym}
\usepackage{amssymb}
\usepackage{amsmath}
\usepackage{latexsym}

\newtheorem{theorem}{Theorem} 
\newtheorem{corollary}[theorem]{Corollary} 
\newtheorem{proposition}[theorem]{Proposition} 
 
\newtheorem{lemma}[theorem]{Lemma} 
 
\newtheorem{observ}[theorem]{Observation}

\newenvironment{proof}{{\bf Proof:\ }}{\hfill $\Box$\linebreak\vskip2mm}

\newcommand{\goes}{\mathop{\longrightarrow}\limits}

\def\choose#1#2{\begin{pmatrix} #1 \\ #2 \end{pmatrix}}

\newcommand{\vu}{\vec{u}}
\newcommand{\vv}{\vec{v}}

\def\Prob#1{{\bf P}\left(#1\right)}
\def\Expect#1{{\bf E}\left(#1\right)}

\def\OPT{\textsf{OPT}}
\def\sat{\textsf{sat}}
\let\vr=\varrho
\let\kp=\kappa
\let\al=\alpha
\let\vf=\varphi
\let\gm=\gamma

\def\vc#1#2{#1 _1\zd #1 _{#2}}
\def\zd{,\ldots,}
\def\Pl{{\Phi^{\sf plant}(n,\vr n)}}
\def\ov#1{{\overline{#1}}}
\newcommand{\Esc}{\mathcal{E}}
\let\lb=\linebreak

\begin{document}
\title{Phase transition for Local Search on planted SAT}

\author{Andrei A.\ Bulatov \hspace*{2cm} Evgeny S. Skvortsov\\
Simon Fraser University\\ 
\{abulatov,evgenys\}@cs.sfu.ca}
\date{}
\maketitle

\begin{abstract}
The Local Search algorithm (or Hill Climbing, or Iterative Improvement) is one of the simplest heuristics to solve the Satisfiability and Max-Satisfiability problems. It is a
part of many satisfiability and max-satisfiability solvers, where it is used
to find a good starting point for a more sophisticated heuristics, and to
improve a candidate solution. In this paper we give an analysis of Local
Search on random planted 3-CNF formulas. We show that if there is $\kappa<\frac76$
such that the clause-to-variable  ratio is less than $\kappa\ln n$ ($n$ is the number
of variables in a CNF) then Local Search whp does not find a satisfying assignment,
and if there is $\kappa>\frac76$ such that the clause-to-variable  ratio is greater
than $\kappa\ln n$ then the local search whp finds a satisfying
assignment. As a byproduct
we also show that for any constant $\vr$ there is $\gm$ such that Local
Search applied to a random (not necessarily planted) 3-CNF with
clause-to-variable ratio $\vr$ produces an assignment that satisfies at least
$\gm n$ clauses less than the maximal number of satisfiable clauses.
\end{abstract}

\section{Introduction}

A CNF formula over variables $\vc xn$ is a conjunction of clauses $\vc cm$ where
each clause is a disjunction of one or more literals. A formula is said to be a
$k$-CNF if every clause contains exactly $k$ literals. In the problem $k$-SAT
the question is, given a $k$-CNF, decide if it has a satisfying assignment
(find such an assignment for the search problem). In the MAX-$k$-SAT problem
the goal is to find an assignment that satisfies as many clauses as possible.
The problem $k$-SAT for $k\ge3$ is one of the first problems proved to
be NP-complete problems and serves as a model problem for many
algorithm and complexity concepts since then.
In particular, H\aa stad \cite{Hastad01:inapproximability} proved that
the MAX-$k$-SAT problem is  NP-hard to approximate within ratio better
than 7/8. These worst case hardness results motivate the study of the
typical case complexity of those problems, and a quest for
probabilistic or heuristic algorithms with satisfactory  performance,
in the typical case. In this paper we analyze the
performance of one of the simplest algorithms for (MAX-)$k$-SAT, the
Local Search algorithm, on random planted instances.

\paragraph{The distribution.}
Let us start with planted
instances. One of the most natural
and well studied probability distributions on the set of 3-CNFs is
the uniform distribution $\Phi(n,m(n))$ on the set of 3-CNFs with a
given clauses-to-variables ratio \cite{Franco83:probabilistic}. It can
be constructed and sampled as follows. Fix the number $m=m(n)$ of 3-clauses as
a function of the number $n$ of variables. The elements of
$\Phi(n,m(n))$ are 3-CNFs generated by selecting $m=m(n)$ clauses over variables
$\vc xn$. Clauses are chosen uniformly at random from the set of possible clauses,
and so the probability of every 3-CNF from $\Phi(n,m(n))$ is the same. An important
parameter of such CNFs is the \emph{clause-to-variable ratio}, $\frac mn$, or
\emph{density} of the formula. We will use the density of a 3-CNF rather than the
number of clauses, and so we write $\Phi(n,\vr n)$ instead of $\Phi(n,m(n))$.
Density can also be a function of $n$.

However, the typical case complexity for this distribution is not very
interesting except for a very narrow range of densities. The reason is that the random 3-SAT under this distribution demonstrates a
sharp satisfiability threshold in the density \cite{Achlioptas99:threshold}. A random
3-CNF with density below the threshold (estimated to be around 4.2) is
satisfiable whp (with high probability,
meaning that the probability tends to 1 as $n$ goes to infinity), and a 3-CNF with
density above the threshold is unsatisfiable whp. Therefore the
trivial algorithm outputting yes or no by just counting the density of
a 3-CNF gives a right answer to 3-SAT whp. For more results on the threshold
see
\cite{Crawford96:crossover,Dubois00:typical,Achlioptas01:lower,Kaporis02:greedy}.
It is also
known that, as the density grows, the number of clauses satisfied by a random
assignment differs less and less from the maximal number of satisfiable clauses.
If density is \emph{infinite} (meaning it is an unbounded function of $n$), then
whp this difference becomes negligible, i.e.\ $o(n)$. Therefore, distribution
$\Phi(n,\vr n)$  is not very interesting for MAX-3-SAT, at least when
density is large, as one can get whp a very good approximation just by
checking a random assignment.

A more interesting and useful distribution is obtained from
$\Phi(n,\vr n)$ by conditioning
on satisfiability: such distribution is uniform and its elements are
the satisfiable 3-CNFs. Then the problem is to find or approximate a satisfying
assignment knowing it exists. Unfortunately, to date there are no techniques to
tackle such problems (see, e.g.,
\cite{Ben-Sasson02:planted,Chen03:above}), particularly, to sample the satisfiable
distribution. A
good approximation for such a distribution is the
planted distribution $\Pl$, which is obtained from $\Phi(n, \vr n)$ by conditioning
on satisfiability by a specific ``planted'' assignment. To construct
an element of a planted distribution we select an assignment of a set
of $n$ variables and then uniformly at random include $\vr n$ clauses
satisfied by the assignment selected. Some attempts have been made to
define a better approximation of the satisfiable distribution, see,
e.g.\ \cite{Krivelevich08:satisfiable}, however, the analysis of such
distributions is difficult and it is not clear if they are closer to
the distribution sought.

Another interesting feature of the planted distribution is that there
is a hope that it is possible to design an algorithm that solves
all planted instances whp. Some candidate algorithms were suggested in
\cite{Ben-Sasson02:planted,Flaxman03:spectral,Krivelevich06:polynomial}. Algorithm
from \cite{Flaxman03:spectral} and \cite{Krivelevich06:polynomial} use different
approaches to solve planted 3-SAT of high density. Experiments show that the
algorithm from \cite{Skvortsov07:random} achieves the goal, but a rigorous analysis
of this algorithm is not yet made. For a wider survey on SAT
algorithms the reader is referred to \cite{Mitchell05:primer,Braunstein05:survey}.

\paragraph{The algorithm.}
The Local Search algorithm (LS) is one of the
oldest heuristics for SAT that has been around since the
eighties. Numerous variations of  this method have been proposed since
then, see, e.g., \cite{Gu92:efficient,Selman92:new}. We study one of the
most basic versions of LS, which, given a CNF, starts with a random assignment
to its variables, and then on each step chooses at random a variable
such that flipping this variable increases the number of satisfied
clauses, or stops if such a variable does not exist. Thus LS finds a
random local optimum accessible from the initial assignment.

LS has been studied before. The worst-case performance of pure LS is not very good:
the only known lower bound for local optima of a $k$-CNF is $\frac k{k+1}m$ of clauses satisfied, where $m$ is
the number of all clauses \cite{Hansen90:algorithms}. In \cite{Koutsoupias92:greedy}, it
is shown that if density of 3-CNFs is linear, that is,  $m=\Omega(n^2)$, then LS solves
whp a random planted instance. Finally, in \cite{Bulatov06:efficiency}, we gave an
estimation of the dependence of the number of clauses LS typically satisfies and
the density of the formula.

Often visualization of the number of clauses satisfied by an
assignment is useful: Assignments can be thought of as points of a landscape,
and the elevation of a point corresponds to the number of clauses
unsatisfied, the higher the point is, the less clauses it
satisfies. It is suspected that `topographic' properties of such a
landscape are responsible for many complexity properties of
satisfiability instances. For example, it is believed that the
hardness of random CNFs whose density is close to the satisfiability
threshold is due to the geometry of the satisfying assignments. They
tend to concentrate around several centers, that make converging to a
solution more difficult
\cite{Braunstein05:survey,Mezard05:clustering}. As we shall see the
performance of LS is closely related to geometric properties of the
assignments, and so we hope that the study of LS may lead to a better
understanding of those properties.

The behavior of other SAT/MAXSAT algorithms have been studied before. For example, the random walk has been analyzed in \cite{Papadimitriou91:selecting} and then in \cite{Alekhnovich07:upper}. A message passing type algorithm, Warning Propagation, is studied in \cite{Feige06:convergence}.

\paragraph{Our contribution.}
We classify the performance of LS for all densities higher than an arbitrary constant.
In particular, we demonstrate that LS has a threshold in its
performance. The main result  is the following theorem.

\begin{theorem}\label{the:main}
(1) Let $\vr\ge \kp\cdot\ln n$, and $\kp>\frac76$. Then the local search whp finds
a solution of an instance from $\Pl$.\\[1mm]
(2) Let $c \le \vr\le \kp\cdot\ln n$, $c$ a constant, and $\kp<\frac76$. Then the local search whp does
not find a solution of an instance from $\Pl$.
\end{theorem}

To prove part (1) of the theorem~\ref{the:main} we show that under
those conditions all the local optima of a 3-CNF whp are either
satisfying assignments, that is, global optima, or obtained by
flipping almost all the values of planted solution, and so are located
on the opposite side of the set of assignments. In the former case
LS finds a satisfying assignment, while whp it does not reach the local
optima of the second type. We also show
that that for any constant density $\vr$ there is $\gm$ such that
the assignment produced by LS on an instance from $\Pl$ or $\Phi(n,\vr n)$
satisfies at least $\gm n$ clauses less than the maximal number of satisfiable
clauses. Unfortunately, it is somewhat difficult to run computational
experiments on CNFs of infinite density, as in order to have $\log n$
sufficiently large $n$ must be prohibitively big. However, experiments
we were able to conduct agree with the results.

Another region where LS can find a solution of the random planted 3-CNF is the case of very low density. Methods similar to Lemma~\ref{lem:close-isolated} and Theorem~\ref{the:finite-density} show that this low density transition happens around $\varrho \approx n^{-1/4}$. However, we do not go into details here.

Usually the main difficulty of analysis of algorithms for random SAT
is to show that as an algorithm runs, some kind of
randomness of the current assignment is kept. This property allows one
to use `card games', Wormald's theorem, and differential equations as in
\cite{Achlioptas01:lower,Bulatov06:efficiency}, or relatively simple probabilistic
constructions, such as martingales, as in
\cite{Alekhnovich07:upper}. For LS randomness cannot be assumed after
just a few iterations of the algorithm, which makes its analysis more
difficult. This is why the most difficult part of the proof is to
identify to which extent assignments produced by LS as it runs remain
random, while most of the probabilistic computations are fairly standard.

The paper is organized as follows. After giving several necessary definitions in
Section~2, we prove in Section~3, that above the threshold established
in Theorem~\ref{the:main} planted 3-CNFs do not have local optima that
can be found  by LS, other than satisfying assignments. In Section~4
we show that below the  threshold there are many such optima, and that
LS necessarily gets stuck into one of them. 

\section{Preliminaries}

\paragraph{SAT.}
A 3-CNF is a conjunction of \emph{3-clauses}. As we consider only 3-CNFs, we
will always call them just clauses. Depending on the number
of negated literals, we distinguish 4 types of clauses:
$(-,-,-),(+,-,-)$, $(+,+,-)$, and
$(+,+,+)$. If $\vf$ is a 3-CNF over variables $\vc xn$, an \emph{assignment} of these
variables is a Boolean $n$-tuple $\vu=(\vc un)$, so the value of $x_i$ is $u_i$.
The \emph{density} of a 3-CNF $\vf$ is the number $\frac mn$ where $m$ is the
number of clauses, and $n$ is the number of variables in $\vf$.

The \emph{uniform} distribution of 3-CNFs of density $\vr$ (density may be a
function of $n$), $\Phi(n,\vr n)$ is the set of all 3-CNFs containing $n$ variables
and $\vr n$ clauses equipped with the uniform probability distribution on this set.
To sample a 3-CNF accordingly to $\Phi(n,\vr n)$ one chooses uniformly and
independently $\vr n$ clauses out of $2^3\choose n3$ possible clauses. Thus,
we allow repetitions of clauses, but not repetitions of variables within a clause.
\emph{Random 3-SAT} is the problem of deciding the satisfiability of a 3-CNF
randomly sampled accordingly to $\Phi(n,\vr n)$. For short, we will call such a
random formula a 3-CNF from $\Phi(n,\vr n)$.

The \emph{uniform planted} distribution of 3-CNF of density $\vr$ is constructed as
follows. First, choose at random a Boolean $n$-tuple $\vu$, a \emph{planted}
satisfying assignment. Then let $\Phi^{\sf plant}(n,\vr n,\vu)$ be the
uniform probability distribution over the set of all 3-CNFs over
variables $\vc xn$ with density $\vr$ and such that $\vu$ is a   satisfying assignment.
For our goals we can always assume that $\vu$ is the all-ones tuple, that is a 3-CNF
belongs to $\Phi^{\sf plant}(n,\vr n,\vu)$  if and only if it contains no clauses
of the type $(-,-,-)$. We also simplify the notation $\Phi^{\sf plant}(n,\vr n,\vu)$
by $\Pl$. To sample a 3-CNF accordingly to $\Pl$ one chooses uniformly and
independently $\vr n$ clauses out of $7\choose n3$ possible clauses of types
$(+,-,-),(+,+,-)$, and $(+,+,+)$. \emph{Random Planted 3-SAT} is the problem
of deciding the satisfiability of a 3-CNF from $\Pl$.

The problems \emph{Random MAX-3-SAT} and \emph{Random Planted MAX-3-SAT} are
the optimization versions of Random 3-SAT and Random Planted 3-SAT. The goal
in these problems is to find an assignment that satisfies as many clauses as
possible. Although the two problems usually are treated as maximization problems,
it will be convenient for us to consider them as problems of minimizing the number
of unsatisfied clauses. Since we always evaluate the absolute error of our algorithms,
not the relative one, such transformation does not affect the results.

\paragraph{Local search.}
A formal description of the Local Search algorithm (LS) is given in
Fig.~\ref{fig:LS}.
\begin{figure}[t]
\begin{tabbing}
{\sc Input:} 3-SAT formula $\vf$ over variables $\vc xn$.\\
{\sc Output:} Boolean $n$-tuple $\vv$, which is a local minimum of $\vf$.\\
{\sc Algorithm:} \\
{\bf choose} uniformly at random a Boolean $n$-tuple $\vu$\\
{\bf let} $U$ be the set of all variables $x_i$
such that the number of
clauses that can be made satisfied\\
\ \ \ \  by flipping the value of $x_i$ is strictly greater
than the number of those made unsatisfied\\
{\bf while} $U$ is not empty\\
\ \ \ \= {\bf pick} uniformly at random a variable $x_j$ from $U$\\
\> {\bf change} the value of $x_j$\\
\> {\bf recompute} $U$
\end{tabbing}
\caption{Local Search}
\label{fig:LS}
\end{figure}
Observe that LS stops when reaches a local minimum of the
number of unsatisfied clauses.

Given an assignment $\vu$ and a clause $c$ it will be convenient to
say that $c$ \emph{votes} for a variable $x_i$ to have value 1 if $c$
contains literal $x_i$ and its other two literals are unsatisfied. In
other words if either (a)
$\vu$ assigns $x_i$ to 0, $c$ is not satisfied by $\vu$,
  and it will be satisfied if the value of $x_i$ is changed, or
(b) the only literal in $c$ satisfied by $\vu$ is $x_i$. Similarly, we say that $c$ votes for $x_i$ if $c$ contains the negation of $x_i$ and its other two literals are not satisfied.
Using this terminology we can define set $U$ as the set of all variables
such that the number of votes received to change the current value is
greater than the number of those to keep it.

\paragraph{Random graphs.}
Probabilistic tools we use are fairly standard and can be found in the
book \cite{Alon00:probabilistic}.

Let $\vf$ be a 3-CNF with variables $\vc xn$. The \emph{primal graph}
$G(\vf)$ of $\vf$ is the graph with vertex set $\{\vc xn\}$ and edge set
$\{x_ix_j\mid \hbox{literals containing } x_i,x_j\hbox{ appear in the same clause}\}$.  The \emph{hypergraph}
$H(\vf)$ \emph{associated with} $\vf$ is a hypergraph, whose vertices are the variables
of $\vf$ and the edges are the 3-element sets of variables belonging to the same
clause. Note that if $\vf\in\Pl$, then $H(\vf)$ is a random 3-hypergraph with $n$
vertices and $\vr n$ edges, but $G(n)$ is not a random graph.

We will need the following properties that a graph $G(\vf)$ of not too high density has.

\begin{lemma}\label{lem:random}
Let $\vr<\kappa\ln n$ for a certain constant $\kappa$, and let $\vf\in\Pl$.

(1) For any $\al<1$, whp all the subgraphs of $G(\vf)$ induced by at most
$O(n^\al)$ vertices have the average degree less than~5.

(2) The probability that $G(\vf)$ has a vertex of degree greater than
$\ln^2 n$ is $o(n^{-3})$.
\end{lemma}

\begin{proof}
(1) This part of the lemma is very similar to Proposition~13 from \cite{Feige06:convergence},
and is proved in a similar way. 
Let $S$ be a fixed set of variables with $|U|=\ell$. The number of 3-element
sets of variables that include 2 variables from $U$ is bounded from above by
$$
\choose{\ell}2 (n-2)\le \frac12 \ell^2n.
$$
For each of them the probability that this set is the set of variables of one of the
random clauses chosen for $\vf$ (we ignore the type of the clause) equals
$$
\frac{\kp n\ln n}{\choose n3}=\frac{6\kp\ln n}{(n-1)(n-2)}.
$$
Thus, the probability that $2\ell$ of them are included as clauses is at most
$$
\choose{\frac12 \ell^2n}{2\ell}\left(\frac{6\kp\ln n}{(n-1)(n-2)}\right)\le
\left(3e\kp\cdot\frac{\ell\ln n}n\right)^{2\ell}.
$$
Let $d=e(3e\kp)^2$. Using the union bound, the probability that there
exists a required set $U$ with at most $n^\al$ variables is at most
\begin{eqnarray*}
\lefteqn{
\sum_{\ell=2}^{n^\al}\choose nk\left(\sqrt{\frac de}\frac{\ell\ln n}n\right)^{2\ell}}\\
&\le& \sum_{\ell=2}^{n^\al}\left(\frac{ne}\ell\cdot
\frac de\cdot\frac{\ell^2\ln^2 n}{n^2}\right)^\ell\\
&\le& \sum_{\ell=2}^{n^\al}\left(d\frac{n^\al\ln^2 n}n\right)^\ell\\
&=& (dn^{\al-1}\ln^2 n)^2\frac{1-(dn^{\al-1}\ln n)^{\ell-1}}{1-dn^{\al-1}\ln n}\\
&=& O(n^{2\al-2}\ln^4 n).
\end{eqnarray*}

(2) The probability that the degree of a fixed vertex is at least $\ln^2n$ is bounded from above by 
$$
\left(\frac1n\right)^{\ln^2n}\choose{3\kp n\ln n}{\ln^2 n}\le
n^{-\ln^2 n}\left(\frac{3e\kp n\ln n}{\ln^2 n}\right)^{\ln^2 n}=
\left(\frac{3e\kp}{\ln n}\right)^{\ln^2 n},
$$
where $n^{-\ln^2n}$ is the probability that some particular $\ln^2n$
random clauses include $x$, and $\choose{3\kp n\ln n}{\ln^2 n}$ is the
number of $\ln^2n$-element sets of clauses. Then it is not hard to see
that
$$
n\left(\frac{3e\kp}{\ln n}\right)^{\ln^2 n}\goes0,
$$
as $n$ goes to infinity.
\end{proof}

Several times we need the following corollary from Azuma's inequality
for supermartingales (see Lemma~1 from \cite{Wormald95:differential}).

\begin{observ}\label{obs:azuma_c}
(1) Let $Y_t$ be a
supermartingale such that $\Expect{Y_{t+1} | Y_{t}} \leq Y_t$ and
$|Y_{t+1}-Y_{t}|<c$ for some $c$. Then 
$
\Prob{Y_t - Y_0 \geq bc} \leq e^{- \frac{b^2}{2t}},
$
for any $b>0$.

(2) This inequality implies that if $\Expect{Y_{t+1}|Y_t} < Y_t - d$ and
$|Y_{t+1}-Y_{t}|<c\leq 1$ then the process $Z_t = Y_t - dt$ is a
supermartingale and we have the following inequality
\begin{equation}
\Prob{Y_t - Y_0 \geq bc} = \Prob{Z_t-Z_0\le \left(b+\frac{dt}c\right)}\le
e^{-\frac{(b + dt)^2}{2tc^2}}
\leq  e^{- bd}.   \label{eq:azuma_c}
\end{equation}
\end{observ}

The following lemma is a simple corollary of Chernoff
bound. 

\begin{lemma}\label{lem:vector-chernoff} 
Let $r,s$ be integers, $\theta<1$ a positive real, and let $\alpha_1, \dots, \alpha_r, \beta_1, \dots, \beta_s$ be some real
constants. There are constants $\lambda$ and $C$ such that we have
\begin{equation}
\Prob{X>Y} < C e^{-\lambda \Expect{Y}} \label{eq:vector-chernoff}
\end{equation}
for any random variables $X$ and $Y$ such that $\Expect{X} < \theta \Expect{Y}$ and 
$X = \sum\limits_{i = 0}^r
\alpha_i X_i, Y = \sum\limits_{i = 0}^s \beta_i Y_i$ for some binomial random variables $X_1, \dots, X_r, Y_1, \dots, Y_s$.
\end{lemma}

\begin{proof}
Let $\xi = \frac{1 - \theta}{(r + s) \max(\max(\alpha_i), \max(\beta_i))}$.
It is easy to see that event $X>Y$ implies occurrence of at least one
of the events from the set
$$
\mathcal{S} = \{\{X_i \geq \Expect{X_i} +
\xi \Expect{Y} \}_{i \in \{0, \dots, r\}}, \{ Y_i \leq \Expect{Y_i} -
\xi \Expect{Y} \}_{i \in \{ 0, \dots, s\}}\}.
$$
Indeed, inequality $X<Y$ can be derived from  inequalities, opposite
to the ones in $\mathcal{S}$ and $\Expect{X} < \theta \Expect{Y}$.

Application of Chernoff bound gives us inequalities
\begin{eqnarray*}
\Prob{|X - \Expect{X_i}| > \xi \Expect{Y}} &<& e^{-\Expect{X_i} \xi^2
\left(\frac{\Expect{Y}}{\Expect{X_i}} \right)^2/3} \leq e^{-\xi^2
\Expect{Y} \theta^{-2} / 3},\\
\Prob{|Y - \Expect{Y_i}| > \xi \Expect{Y}} &<& e^{-\Expect{Y_i} \xi^2
\left(\frac{\Expect{Y}}{\Expect{Y_i}} \right)^2/3} \leq e^{-\xi^2
\Expect{Y} / 3}.
\end{eqnarray*}

Thus if we set $\lambda = \xi^2/3$, $C = r + s$ then using union bound
we can conclude that inequality (\ref{eq:vector-chernoff}) holds.
\end{proof}

\section{Success of Local Search}

In this section we prove the first statement of the
Theorem~\ref{the:main}(1). This will be done as follows. First, we
show that if a 3-CNF has high density, that is, greater than $\kappa
\log n$ for some $\kappa>\frac76$ then whp all the local minima that
do not satisfy the CNF --- we call such minima
\emph{proper} --- concentrate very far from the planted
assignment. This is the statement of Proposition~\ref{pro:minima1}
below. Then we use Lemma~\ref{lem:no-zero-flood} to prove that
starting from a random assignment LS whp does not go to that remote
region. Therefore the algorithm does not get stuck to a local minimum
that is not a solution.

Several times we will need the following observation that can be
checked using the inequality\lb $\choose
n\ell\le\left(\frac{ne}\ell\right)^\ell$.
For any $n$, $\gm$, and $\al$ with $0<\al<1$
\begin{equation}\label{equ:estimate}
\choose{n}{\gamma n^\alpha}\le e^{(1-\alpha) \gamma n^{\alpha} \ln n - \gamma
n^\alpha \ln\gamma + \gamma n^\alpha}.
\end{equation}

We need the following two lemmas. Recall that the planted solution is the all-ones one.

\begin{lemma}\label{lem:no-zero-flood}
Let $\vr \ge \kp\ln n$ for some constant $\kp$, and let  constants $q_0, q_1$
be such that $q_0<q_1$. Whp any assignment with $q_0n$ zeros
satisfies more clauses than any assignment with $q_1n$ zeros.
\end{lemma}

\begin{proof}
Let $\vu,\vv$ be some vectors with $q_0n$ and $q_1n$ zeros, respectively.
Let $c$ be a random clause, then (1) with probability $\frac{1}{7}$ all its
literals are positive, (2) with probability $\frac{3}{7}$ two literals are
positive and similar (3) with probability $\frac{3}{7}$ one literal is positive.
The probabilities that the clause is satisfied by $\vu$ in these cases
are $(1-q_0)^3, (1-q_0)^2q_0$ and $(1-q_0)q_0^2$, respectively. Hence
the total probability of a clause to be satisfied by $\vu$ equals
$\frac{(1-q_0)^3 + 3 (1-q_0)^2q_0 + 3 (1-q_0) q_0^2}{7} = \frac{1 - q_0^3}{7}$.
A similar result holds for $\vv$.
Thus the expectation of the number of clauses satisfied by $\vu$ and $\vv$ in a
random formula equals $\frac{1-q_0^3}{7} \kp n \ln n$ and $\frac{1-q_1^3}{7}
\kp n \ln n$ respectively, thus applying lemma~\ref{lem:vector-chernoff} we
conclude that
$$
\Prob{\mbox{ $\vu$ satisfies less than $\frac{1 - (q_0^3 + q_1^3)/2}{7}
\kp n \ln n$ clauses}} < e^{-\lambda' n \ln n},
$$
for some $\lambda'>0$. 
There are $2^n$ assignments, hence, application of the union bound finishes proof
of the lemma.
\end{proof}

\begin{lemma}\label{lem:close-minima}
Let $\vr\ge \kp\ln n$ for some $\kp$ (not necessarily
$>\frac76$). There is $\al<1$ such that for $\vf\in\Pl$
whp for  any proper local minimum $\vu$ of $\vf$ the number of
variables assigned to~0 by $\vu$ is either less than $n^\al$, or
greater than $\frac {9n}{10}$.
\end{lemma}

\begin{proof}
Let $M, |M|=\ell$ be the set of all variables that $\vu$ assigns to 0.
Let $\mathcal{B}^{each}_M$ be event ``for every $x_i \in M$ the number of clauses
voting for $x_i$ to be 1 is less than or equal to the number of clauses voting
for $x_i$ to be 0''. Since $\vu$ is a local minimum, $\mathcal{B}^{each}_M$ is the case for $\vu$. It is easy to see that event $\mathcal{B}^{each}_M$
implies event $\mathcal{B}^{all}_M =$ ``the total number of votes given by clauses
for variables in $M$ to be 1 is less than or equal to the total number of votes
given by clauses for variables in $M$ to be 0''. To bound the probability of
$\mathcal{B}^{each}_M$ we will bound the probability of $\mathcal{B}^{all}_M$.

Let $c$ be a random clause. It can contribute from 0 to 3 votes for variables in
$M$ to be one and 0 or 1 vote for them to remain zero. Let us compute, for example,
the probability that it contributes exactly two votes for variables in $M$
to become one. It happens if
$c$ is of type $(+,+,-)$, both its positive variables are in $M$ and the negative variable
is outside of $M$. Probability of this event is $\frac{3}{7} \ell^2 n^{-2} (1 - \ell/n)$.
So the expectation of the number of clauses voting for exactly 2 variables in $M$ to be
1 is $\frac{3}{7} \ell^2 n^{-1} (1 - \ell/n) \kappa \ln n$. The expectations of the numbers of
clauses voting for three and one variables to be 1 are $\frac{1}{7} \ell^3 n^{-2} \kappa \ln n$ and $\frac{3}{7} (1-\frac{\ell}{n})^2 \ell \kappa \ln n$, respectively.

A clause votes for a variable in $M$ to remain 0 if its type is $(+,-,-)$, one
of its negative literals is not in $M$, and two other literals are in $M$, or if its type
is $(+,+,-)$ and all the variables in it belong to $M$. Thus the expectation of the
number of clauses voting for variables in $M$ to remain 0 is $\frac{3}{7} \kappa
\ln n \left( 2 \ell^2 n^{-1} (1-\ell/n) +\ell^3 n^{-2}\right)$.

Hence the expectation of the number of votes for variables in $M$ to flip
equals
$$
\Expect{\mbox{votes for a flip}} = \kappa\ln n \times
\left(3\cdot\frac{1}{7} \ell^3 n^{-2} +2\cdot\frac{3}{7} \ell^2 n^{-1} (1 - \ell/n) + 
1\cdot\frac{3}{7} \ell (1 - \ell/n)^2 \right)
$$
and expectation of the
number of votes for variables in $M$ to remain 0 equals
$$
\Expect{\mbox{votes for status quo}} = \kappa \ln n \times\left(
\frac{6}{7} \ell^2n^{-1} (1-\ell/n) + \frac{3}{7} \ell^3 n^{-2}\right).
$$

If $\ell<\frac9{10}n$ then 
\begin{eqnarray*}
\frac{\Expect{\mbox{votes for status quo}}}{\Expect{\mbox{votes for a flip}}} &=&
\frac{6\ell(n-\ell)+3\ell^2}{6\ell(n-\ell)+3\ell^2+3(n-\ell)^2}=
1-\frac{3(n-\ell)^2}{6\ell(n-\ell)+3\ell^2+3(n-\ell)^2}\\
&<& 1-\frac{3\cdot\frac1{100}n^2}{12n^2}=1-\frac1{400}.
\end{eqnarray*}
Therefore we can apply Lemma~\ref{lem:vector-chernoff} to the votes for and
against 0s and get the following bound
$\Prob{\mathcal{B}^{all}_M} < e^{- \lambda \Expect{\mbox{votes for a flip}}}$ for some $\lambda>0$.
Then we can bound number of votes for a flip from below by $\delta \ell
\ln n$ for some constant $\delta$ and we can bound the number of sets $M$
of size $\ell$ as
$$
\# (\mbox{M of size $\ell$}) = \choose{n}{\ell} \leq
\left(\frac{ne}{\ell}\right)^\ell = e^{\ell \ln(n/\ell) + \ell}.
$$

Therefore if
$$
\ell \ln (n/\ell) + \ell < \delta \ell \ln n 
$$
then union bound implies that whp there is no set $M$ such that
$\mathcal{B}^{all}_M$ happens. It is easy to see that for $\ell>n^\alpha$
and $\alpha$ that is close enough to 1 the above inequality 
holds, which finishes the proof of the lemma.
\end{proof}

Now suppose that $\vu$ is a proper local minimum of $\vf\in\Pl$.
There is a clause $c\in\vf$ that is not satisfied by $\vu$. Without loss of generality,
let the variables in $c$ be $x_1,x_2,x_3$, and let the variable assigned 0 be $x_1$.
Thus, clause $c$ votes for $x_1$ to be flipped to 1. Since $\vu$ is a local minimum
there must a clause that is satisfied, that becomes unsatisfied should $x_1$ flipped.
We call such a clause a \emph{support} clause for the 0 value of $x_1$. In any support
clause the supported variable is negated, and therefore any support clause has the type
$(+,-,-)$ or $(+,+,-)$.
A variable of a CNF is called \emph{$k$-isolated} if it appears positively in at most
$k$ clauses of the type $(+,-,-)$. The \emph{distance}
between variables of a CNF $\vf$ is the length of the shortest path in $G(\vf)$
connecting them.

\begin{lemma}\label{lem:scarcity}
If $\kappa>\frac{7}{6}$ and $\vr\ge\kappa\ln n$ then for any integers
$d_1,d_2\ge1$ and for a random $\vf\in\Pl$ whp there are no two
$d_1$-isolated variables within distance $d_2$ from each other.
\end{lemma}

\begin{proof}
Let $x$ be some variable. The probability that it is
$d_1$-isolated can be computed as
\begin{eqnarray*}
\Prob{x \mbox{ is $d_1$-isolated}} &=&
d_1\cdot\choose{\kp n\ln n}{d_1} \left( 1- \frac{3}{7n} \right)^{\kappa n\ln n -d_1} \left(\frac3{7n}\right)^{d_1}\\
&\le& d_1(\kp n\ln n)^{d_1}\left( 1- \frac{3}{7n} \right)^{\kappa n\ln n} \left( 1- \frac{3}{7n} \right)^{-d_1}\left(\frac73 n\right)^{-d_1}\\
&\sim& d_1 \left( 1- \frac{3}{7n} \right)^{-d_1} (\frac{7\kp}3\ln n)^{d_1} e^{-\frac{3}{7} \kappa \ln n}\\
&=& O(n^{-\frac{3\kappa}{7}+\varepsilon}), 
\end{eqnarray*}
for any $\epsilon>0$.

By Lemma~\ref{lem:random}(2), the degree of every vertex of $G(\vf)$
whp does not
exceed $\ln^2n$. Hence, there are at most $\ln^{2d_2}n$ vertices at distance
$d_2$ from $x$. Applying the union bound we can estimate the probability
that there is a $d_1$-isolated vertex at distance $d_2$ from $x$ as
$O(\ln^{2d_2}n\cdot n^{-\frac37\kp})$. Finally, taking into account the
probability that $x$ itself is $d_1$-isolated, and applying  the union bound
over all vertices of $G(\vf)$ we obtain that the probability that two
$d_1$-isolated vertices exists at distance $d_2$ from each other can be bounded
from above by
$$
n\cdot O(n^{-\frac{3\kp}7})\cdot O(\ln^{2d_2}n\cdot n^{-\frac37\kp})=
O(\ln^{2d_2}n\cdot n^{1-\frac67\kp}).
$$
Thus for $\kappa>\frac{7}{6}$ whp there are no two such vertices.
\end{proof}

\begin{proposition}\label{pro:minima1}
Let $\vr\ge \kp\cdot\ln n$, and $\kp>\frac76$. Then whp proper
  local minima of a 3-CNF
from $\Pl$ have at most $\frac n{10}$ ones.
\end{proposition}

\begin{proof}
Let $\vf\in\Pl$ be a random planted instance. Suppose that $\vu$
is a proper local minimum that has more than $\frac n{10}$ ones. We use the following observation. Let $c$ be a clause
not satisfied by $\vu$. Then it contains at least one variable $x_i$ that is
assigned to zero by $\vu$. The assignment $\vu$ is a local minimum, so there must
be a clause $c'$ that is satisfied only by $x_i$. Hence, $c'$ is a support clause,
and contains a variable $x_j$ which is assigned to zero by $\vu$. Variables
$x_i$ and $x_j$ are at distance~$1$. Setting $d_1=11$ and $d_2=1$, by
Lemma~\ref{lem:scarcity}, we conclude that one of them is not 11-isolated.

Set $d_1 = 11$, $d_2 = 3$ and consider the set $Z$ of all variables assigned
to zero by $\vu$ that are not 11-isolated. By the observation above this set
is non-empty. On the other hand, by Lemma~\ref{lem:close-minima}, $|Z|$ is
$O(n^\al)$ for some $\al<1$. Consider $x\in Z$. It appears positively in at least
10 clauses of the type $(+,-,-)$. Each of these clauses is either unsatisfied or
contains a variable assigned to 0. Suppose there are $k$ unsatisfied clauses among them. Since
$\vu$ is a local minimum, to prevent $x$ from flipping, $x$ must be supported by at
least $k$ support clauses, each of which contains a variable assigned to~0. Thus, at
least 6 neighbors of $x$ in $G(\vf)$ are assigned to~0. Any two neighbors of $x$ are at
distance 2. By Lemma~\ref{lem:scarcity} at least 5 of
the neighbors assigned to~0 are not 11-isolated, and therefore belong to $Z$. Thus
the subgraph induced by $Z$ in $G(\vf)$ has the average degree greater than 5,
which is not possible by Lemma~\ref{lem:random}(1).
\end{proof}

Now we are in a position to prove statement (1) of
Theorem~\ref{the:main}.

\begin{proof}[of Theorem~\ref{the:main}(1)]
By Lemma~\ref{lem:no-zero-flood}
for a $\vf\in\Pl$ whp any assignment with $dn$ variables equal to~1, where
$\frac13\le d\le\frac23$, satisfies more clauses than any assignment with
$\frac n{10}$ equal to~1. Then, whp a random initial assignment for LS assigns
between $\frac13$ and $\frac23$ of all variables to~1. Therefore, whp LS never arrives
to a proper local minimum with less than $\frac n{10}$ variables equal to~1, and,
by Proposition~\ref{pro:minima1}, to any proper local minimum.
\end{proof}

\section{Failure of Local Search}

We now prove statement (2) of Theorem~\ref{the:main}. The overall
strategy is the following. First, we show,
Proposition~\ref{pro:minima2}, that in contrast to the
previous case there are many proper local minima in the close proximity of
the planted assignment. Then we show, Proposition~\ref{pro:failing},
that those local minima are located so that they intercept almost
every run of LS, and thus almost every run is unsuccessful.

We start off with a technical lemma. A pair of clauses $c_1=(x_1,\ov x_2,\ov x_3)$,
$c_2=(\ov x_1,\ov x_4,x_5)$ is called a \emph{cap}
if $x_1,x_5$ are 1-isolated, that is they do not appear in any clause of the type $(+,-,-)$ except for $c_1$ and $c_2$, respectively, and $x_2,x_3$ are not 0-isolated
(see Figure~\ref{fig:cap}(a)). We denote equality $f(n) = g(n) (1 + o(n))$ by $f(n) \sim g(n)$.

\begin{figure}[t]
\centerline{\includegraphics[totalheight=1.5cm,keepaspectratio]{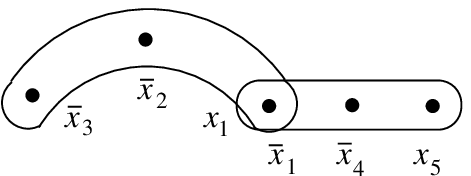}
\hspace{3cm}
\includegraphics[totalheight=2.8cm,keepaspectratio]{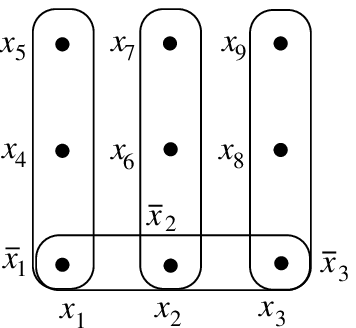}}
\caption{Caps and crowns}
\label{fig:cap}
\end{figure}

\begin{lemma}\label{lem:close-isolated}
Let $n^{-\frac14}<\vr\le \kp\cdot\ln n$, and $\kp<\frac76$. There is $\al$, $0<\al<1$, such that whp a random planted CNF $\vf\in\Pl$ contains at least $n^\al$ caps. 
\end{lemma}

\begin{proof}
The proof is fairly standard, see, e.g.\ the proof of Theorem~4.4.4 in
\cite{Alon00:probabilistic}. We use the second moment method. The result follows from the fact that a cap has properties
similar to the properties of \emph{strictly balanced graphs}, see
\cite{Alon00:probabilistic}. Take some $n$, and let $X$ be a random variable
equal to the number of caps in a 3-CNF $\vf\in\Pl$. Straightforward calculation
shows that the probability that a fixed 5-tuple of variables is a cap is
 $\sim \vr^4 n^{-4-\frac67 \frac\varrho{\ln n}}$. Therefore  $\Expect X\sim \vr^4 n^{1-\frac67 \frac\varrho{\ln n}}$.

Let $S$ be a fixed 5-tuple of variables, say, $S=(x_1,x_2,x_3,x_4,x_5)$, and $A_S$
denote the event that $S$ forms a cap. For any other 5-tuple $T$, the similar event
is denoted by $A_T$, and we write $A_T\asymp A_S$ if these two events are not independent. By Corollary~4.3.5 of \cite{Alon00:probabilistic} it suffices to show that
$$
\Delta^*=\sum_{T\asymp S}\Prob{A_T\mid A_S}=o(\Expect X).
$$

Let $T=(y_1,y_2,y_3,y_4,y_5)$. It is not hard to see that the only cases when
$A_T$ and $A_S$ are not independent and the probability $\Prob{A_T\mid A_S}$ is significantly different from 0 is: $y_1=x_1$ and $\{y_2,y_3\}=\{x_2,x_3\}$,
or $y_1=x_5$ and
$\{y_2,y_3\}=\{x_1,x_4\}$, or $y_5=x_1$ and $\{y_1,y_4\}=\{x_2,x_3\}$, or
$y_5=x_5$ and $\{y_1,y_4\}=\{x_1,x_4\}$.
Then, as before,
it can be found that in each of these cases $\Prob{A_T\mid
A_S}=O(\vr^4 n^{-2-\frac37 \frac\varrho{\ln n}})$.

Finally, 
\begin{eqnarray*}
\Delta^* &=& \sum_{T\asymp S}\Prob{A_T\mid A_S}=n^2\Prob{A_T\mid A_S}=
n^2\cdot O(\vr^4 n^{-2-\frac37 \frac\varrho{\ln n}})\\
&=& O(\vr^4 n^{-\frac37 \frac\varrho{\ln n}})=o(\Expect X).
\end{eqnarray*}
We can choose $\al=1-\frac67\kp$ if $\vr\ge 1$, and $\al=1-4\nu$ if $1>\vr>n^{-\nu}$ for $\nu<\frac14$.
\end{proof}

\begin{proposition}\label{pro:minima2}
Let $\vr\le \kp\cdot\ln n$, and $\kp<\frac76$. Then there is $\al$, $0<\al\le1$,
such that a 3-CNF from $\Pl$ whp has at least $n^\al$
proper local minima.
\end{proposition}

\begin{proof}
Let $c_1=(x_1,\ov x_2,\ov x_3)$, $c_2=(\ov x_1,\ov x_4,x_5)$ be a cap and $\vu$
an assignment such that $u_3=u_5=0$, and $u_i=1$ for all other $i$. It
is straightforward
that $\vu$ is a proper local minimum. By
Lemma~\ref{lem:close-isolated}, there is $\al$
such that whp the number of such minima is at least $n^\al$.
\end{proof}

Before proving Proposition~\ref{pro:failing}, we note that a construction
similar to caps helps evaluate the approximation rate of the local
search in the case of
constant density on planted and also on arbitrary CNFs. A subformula
$c=(x_1,x_2,x_3), c_1=(\ov x_1,x_4,x_5), c_2=(\ov x_2,x_6,x_7),
c_3=(\ov x_3, x_8,x_9)$ is called a \emph{crown} if the variables
$x_1\zd x_9$ do not appear in any clauses other than
$c,c_1,c_2,c_3$ (see Fig.~\ref{fig:cap}(b)). The crown is satisfiable,
but the all-zero assignment is a proper local minimum. For a CNF $\vf$ and an assignment $\vu$ to its variables, by $\OPT(\vf)$ and $\sat(\vu)$ we denote the maximal number of simultaneously satisfiable clauses and the number of clauses satisfied by $\vu$, respectively. 

\begin{theorem}\label{the:finite-density} 
If density $\vr$ is such that $n^{-\nu} \leq \vr \leq \kp \ln n$ for some $\nu<1/4$ and  $\kp<1/27$, then there is $\gm_\vr=\frac1{o(n)}$ such that whp Local Search on a
3-CNF $\vf\in\Phi(n,\vr n)$ ($\vf\in\Pl$) returns an assignment $\vu$ such that
$\OPT(\vf)-\sat(\vu)\ge \gm(\vr)\cdot n$, where $\OPT(\vf)$ denotes
the maximal number of clauses in $\vf$ that can be simultaneously satisfied and $\sat(\vu)$ denotes the number of clauses satisfied by $\vu$.

If $\vr$ is constant then $\gm_\vr$ is also constant.
\end{theorem}

\begin{proof}
As in the proof of Lemma~\ref{lem:close-isolated}, it can be shown that for
$\vr$ that satisfies conditions of this theorem there is $\gm'=\frac{1}{o(n)}$ such that whp a random
[random planted]
formula has at least $\gm'n$ crowns. If $\vr$ is a constant, $\gm'$ is also a constant. For a random assignment $\vu$, whp the
variables of at least $\frac{\gm'}{1024}n$ crowns are assigned zeroes.
Such an all-zero assignment of a crown cannot be changed by the local search.
\end{proof}

Then we move on to proving Proposition~\ref{pro:failing}.

\begin{proposition}\label{pro:failing}
Let $\vr\le \kp\cdot\ln n$, and $\kp<\frac76$. The local search on a 3-CNF from
$\Pl$ whp ends up in a proper local minimum.
\end{proposition}

If $\vr=o(\ln n)$ then Proposition~\ref{pro:failing} follows from Theorem~\ref{the:finite-density}. So in what follows we assume that $\vr>\kp'\cdot\ln n$. The main tool of proving Proposition~\ref{pro:failing} is coupling of local
search (LS)
with the algorithm \textsc{Straight
Descent} (SD) that on each step chooses at random a variable assigned to 0 and
changes its value to 1. Obviously SD is not a practical algorithm, since to apply it
we need to know the solution.
For the purposes of our analysis we modify SD as follows.
At each step SD chooses a variable
at random, and if it is assigned 0 changes its value (see Fig.~\ref{fig:SD}(a)).
The algorithm LS is modified in a similar way (see Fig.~\ref{fig:SD}(b)).

\begin{figure}[t]\label{fig:SD}
\parbox[t]{.4\hsize}{
\begin{tabbing}
{\sc Input:} $\vf\in\Pl$ with the all-ones solution,\\ 
\ \ \ Boolean tuple $\vu$,\\
{\sc Output:} The all-ones Boolean tuple.\\
{\sc Algorithm:} \\
{\bf while} there is a variable assigned 0\\
\ \ \ \= {pick} uniformly at random variable $x_j$ from\\
\>\ \ \  the set of all variables \\
\> {\bf if} $u_j=0$ {\bf then set} $u_j=1$
\end{tabbing}
\centerline{(a)}}
\hskip5mm
\parbox[t]{.5\hsize}{
\begin{tabbing}
{\sc Input:} 3-SAT formula $\vf$, Boolean tuple $\vu$,\\
{\sc Output:} \= Boolean tuple $\vv$, which is local\\
\> minima of $\vf$.\\
{\sc Algorithm:} \\
{\bf while} $\vu$ is not a local minima\\
\ \ \ \= {pick} uniformly at random variable $x_j$ from\\
\> the set of all variables \\
\> {\bf if} the number of
clauses that can be made\\
\> \ \ satisfied by flipping the value of $x_i$ is strictly\\
\> \ \ greater 
than the number of those made unsatisfied\\
\> \ \ {\bf then set} $u_j=\ov u_i$
\end{tabbing}
\centerline{(b)}}
\caption{Straight Descent (a) and Modified Local Search (b)}
\end{figure}

It is easy to see that the vector obtained by SD at step $t$ does not depend on the formula. And since SD treats all variables equally we can make the following

\begin{lemma} \label{lem:randomness}
If $SD$ starts its work at a random vector with $m_0$ ones and after step $t$, $t \leq n - m_0$, it arrives to a vector with $m$ ones, then this vector is selected uniformly at random from all vectors with $m$ ones.
\end{lemma}

\begin{proof}
Let us denote the probability that at step $t$ SD arrives to vector $\vu$, conditional to it starts from a vector with $m_0$ ones, by $\Prob{\vu,t,m_0}$. We prove by induction on $t$ that $\Prob{\vu_1,t,m_0}=\Prob{\vu_2,t,m_0}$ for any $\vu_1,\vu_2$ with $m$ ones. We denote this number by $\Prob{t,m,m_0}$. As the starting vector is random, it is obvious for $t=0$. Then for $t>1$ and any vector $\vu$ with $m$ ones we have
\begin{eqnarray*}
\Prob{\vu,t,m_0} &=& \Prob{\vu,t-1,m_0}\cdot\frac mn + \sum_{\vu'}\Prob{\vu',t-1,m_0} \cdot \frac 1n\\
&=& \Prob{t-1,m,m_0}\cdot\frac mn + \Prob{t-1,m-1,m_0} \cdot \frac{m}n,
\end{eqnarray*}
where $n$ is the number of variables in the formula and $\vu'$ goes over all vectors that can be obtained from $\vu$ by flipping a one into zero. It does not depend on a particular vector $\vu$.
\end{proof} 

We will frequently use the following two properties of the algorithm SD.

\begin{lemma}\label{lem:running-time}
Whp the running time of SD does not exceed
$2n\ln n$.
\end{lemma}

\begin{proof}
For a variable $x_i$ the probability that it is not considered for $t$
steps equals $\left(1 - \frac{1}{n}\right)^t$. So for $t = 2n\ln n$ this probability
equals $\left(1 - \frac{1}{n}\right)^{2 n \ln n} \leq e^{-2 \ln n} = n^{-2}$.
Applying the union bound over all variables we obtain the required statement.
\end{proof}

Given 3-CNF $\vf$ and an assignment $\vu$ we say that a variable $x_i$ is {\em $k$-righteous} if the number of clauses voting for it to be one is greater by at least $k$ than the number of clauses voting for it to be zero. 
Let $\vf\in\Pl$ and $\vu$ be a Boolean tuple. The \emph{ball} of radius $m$ with the
center at $\vu$ is the set of all tuples of the same length as $\vu$ at
Hamming distance at most $m$ from $\vu$. Let $f(n)$ and $g(n)$ be arbitrary functions and $d$ be an integer constant. We say that a set $S$ of
$n$-tuples is \emph{$(g(n), d)$-safe}, if for any $\vu\in S$ the
number of variables that are not $d$-righteous does not exceed $g(n)$. 
A run of SD is said to be \emph{$(f(n), g(n), d)$-safe} if at each step of this
run the ball of radius $f(n)$ with the center at the current assignment is $(g(n), d)$-safe.

\begin{lemma}\label{lem:neighbourhood}
Let $\vr>\kp'\cdot\ln n$ for some $\kp'$. For any constants $\gamma$ and $d$ there is a constant $\alpha_1<1$ such that, for any
$\alpha > \alpha_1$, whp a run of SD on $\vf\in\Pl$ is $(\gm n^\al, n^\al, d)$-safe.
\end{lemma}

\begin{proof}
Consider a run of SD on $\vf\in\Pl$
with a random initial assignment.
If SD starts its work at a tuple with $m_0$ ones, then at step $t$ it has
$m\le m_0+t$ ones. Then by Lemma~\ref{lem:randomness} if at step $t$ the current assignment of SD has $m$ ones then it is drawn
uniformly at random from all vectors with $m$ ones. Event {\em Unsafe} $=$ ``run of SD is not $(\gm n^\al, n^\al, d)$-safe'' is a union of events ``at
step $t$ of SD's run the ball of radius $\gm n^\al$ with the center at
the current assignment is not $(n^\al, d)$-safe''. We will use the union bound to show that probability of {\em Unsafe} is small.

Let $\vu$ be a Boolean $n$-tuple having $pn$ positions filled with 1s.
Since whp the number of 1s in the initial assignment is at least $\frac n3$,
for every step the number of 1s is at least $\frac n3$.
Let $M$ be an arbitrary set of variables with $|M| = n^\alpha$.
We consider events 
$\mathcal{B}^{each}_M =$ ``every variable $x_i \in M$ is not $k$-righteous'' and $\mathcal{B}^{all}_M =$ ``the total number of votes given by clauses
for variables in $M$ to be 1 does not exceed the total number of votes
given by clauses for variables in $M$ to be 0 plus $|M|\cdot k$.''

The same technique as in Lemma~\ref{lem:close-minima}
can be used to show that the probability of $\mathcal{B}^{all}_M$ and consequently
the probability of $\mathcal{B}^{each}_M$ is bounded above by
$
e^{-\lambda' n^\alpha \ln n}
$
for some constant $\lambda'$, not dependent on $\alpha$.
By inequality~(\ref{equ:estimate}), there are at most $\gamma n^\alpha \cdot e^{\gamma (1-\alpha) n^\alpha \ln n  \cdot(1+o(1))}$ distinct assignments in the $\gamma n^\alpha$-neighborhood of SD and \linebreak $e^{n^\alpha (1-\alpha) \ln n (1+o(1))}$ distinct subsets of size $n^\alpha$.
So for $\alpha$ close to 1 the union bound implies that $\mathcal{B}^{each}_M$ whp does not take place for any tuple, any subset of variables at any step which completes the proof of the lemma.
\end{proof}

For CNFs $\psi_1, \psi_2$ we denote by $\psi_1 \wedge \psi_2$ their conjunction. 

\newcommand{\Plpsi}{\Phi^{\tt plant}_\eta(n, \varrho n)}
We will need formulas that obtained from a random formula by adding some clauses in an `adversarial' manner. Following \cite{Krivelevich06:polynomial} we call distributions for such formulas \emph{semi-random}. However, the type of semi-random distributions we need is different from that in \cite{Krivelevich06:polynomial}. Let $\eta<1$ be some constant. 
A formula $\vf$ is sampled according to semi-random distribution $\Phi^{\tt plant}_\eta(n, \varrho n)$ if $\vf = \vf' \wedge \psi$, where $\vf'$ is sampled according to $\Pl$ and $\psi$ contains at most $n^\eta$ clauses and is given by an adversary.

\begin{corollary}\label{cor:run_is_safe}
If $\vf' \in \Phi^{{\tt plant}}_\eta (n, \varrho n )$ 
then for any constants $\gm$ and $d$ there is a constant $\alpha_2<1$ such that for any $\alpha>\alpha_2$ a run of $SD$ on $\vf' \circ \psi$ is whp $(\gm n^{\alpha}, 2n^\al, d)$-safe.
\end{corollary}

\begin{proof}
Let $\alpha_1$ be obtained by application of Lemma~\ref{lem:neighbourhood} to $\vf'$. Let $\alpha_2 = \max(\alpha_1, \eta)$. Then for $\alpha > \alpha_2$ whp run of $SD$ on $\vf'$ is $(\gm n^\al, n^\al, d)$-safe. Since for $n$ large enough $\psi$ contains less than $n^\al$ variables run of $SD$ will be $(\gm n^\al, 2n^\al, d)$-safe on $\vf' \wedge \psi$.
\end{proof}

\pagebreak

\begin{lemma} \label{lem:strip_cross}
Let $(D_0\zd D_l)$ be an integer random process, $d>0$, and let
$L$, $H$ be integer constants such that 
\vspace*{-3mm}
\begin{itemize}
\item[(a)] $D_0 = 0$, $0<L<H$, 
\item[(b)] $|D_{\tau+1} - D_{\tau}| = 1$,
\item[(c)] if $L\leq D_\tau \leq H$ the expectation of $D_{\tau+1}$ conditional to $D_\tau$ satisfies the inequality
$\Expect{D_{\tau+1}|D_\tau} < D_{\tau} - d$
holds.
\end{itemize}
\vspace*{-3mm}
Then the probability that there is $\tau$ such that $D_\tau>H$ is less than
$l \cdot  e^{-d\frac{H - L}{2}}$.\\ \label{eq:mart_bound}
\end{lemma}
\begin{proof}
We define a set of auxiliary processes $D^{\xi}_{\tau}$:
$$
D^{\xi}_{\tau} =
\begin{cases}
L, \ \ \mbox{if }\tau < \xi,\\
D_\tau, \ \ \mbox{if } (\tau\geq \xi),\ (D_\xi = L) \mbox{ and }
(D_\zeta \geq L), \mbox{ for all $\zeta \in \{\xi, \dots,  \tau$\}}),\\
D_\zeta - d (\tau - \zeta)
, \ \ \mbox{if $\tau>\xi$, $D_\xi = L$, and $\zeta \in
\{\xi, \dots, \tau\}$ is the least such that $D_\zeta<L$},\\
L - d (\tau - \xi) ,\ \
\mbox{otherwise, i.e., $D_\xi\neq L$ and $\tau\geq \xi$}.
\end{cases}
$$
The processes $D^0_\tau, \dots, D^l_\tau$ are designed
so that every $D^\xi_\tau$ for $\tau\geq \xi$ satisfies inequality
$\Expect{D^\xi_{\tau+1}|D^\xi_\tau} \leq D^\xi_\tau - d$. Indeed, suppose that $\tau\ge\xi$. If $D_\xi\ne L$ then 
$$
\Expect{D^\xi_{\tau+1}|D^\xi_\tau}=L-d(\tau+1-\xi) =(L-(\tau-\xi)-d = D^\xi_\tau - d.
$$ 
Let $D_\xi=L$. If $D_\zeta\ge L$ for all $\zeta\{\xi\zd \tau\}$ then $D^\xi_\tau=D_\tau$, $D^\xi_{\tau+1}=D_{\tau+1}$, and the result follows from the assumption $\Expect{D_{\tau+1}|D_\tau}<D_{\tau}-d$. If there is $\zeta\in\{\xi\zd \tau\}$ with $D_\zeta<L$ then
$$
\Expect{D^\xi_{\tau+1}|D^\xi_\tau}=\Expect{D^\xi_{\tau+1}|D_\zeta}=D_\zeta-d(\tau+1-\zeta)= (D_\zeta-d(\tau-\zeta))-d=D^\xi_\tau-d.
$$

By Azuma's
inequality~(\ref{eq:azuma_c}) for each $\xi$ the probability of the event
``there exists $\tau$ such that $D_\tau^\xi = H$'' is less than $e^{-(H-L)d}$.

On the other hand let $D_\tau > L$ and $\xi$ be equal to the number of the
most recent step for which $D_\xi = L$. It is easy to see that
$D_\tau = D^\xi_\tau$. Thus if at some step $D_\tau = H$ then there is
$\xi<\tau$ such that $D^\xi_\tau = H$. Using the union bound we get the required
inequality.
\end{proof}

\begin{lemma}\label{lem:coupling1}
Let $\vr>\kp'\cdot\ln n$ for some $\kp'$. Let $\vf$ be a random 3-CNF sampled according to distribution $\Plpsi$ such that run of $SD$ on $\vf$ is whp $(\gm_1 n^\al, \gm_2 n^\al, 1)$-safe for some constants $\gm_1, \gm_2$ with $\gm_1>3 \gm_2$.
Let $\vu_d(m),\vu_l(m)$ denote the pair of assignments produced by the
pair of processes (SD,LS) on step $m$.
For any $t$, whp the Hamming distance between $\vu_d(t)$ and $\vu_l(t)$ does
not exceed $\gm_1 n^\al$.
\end{lemma}

\begin{proof}
Let $N_t$ be the set of tuples at Hamming distance at most $\gm_1 n^\al$ from
$\vu_d(t)$, and $\Esc$ be event ``$\vu_l(t)\not\in N_t$ for some $t$''.
LS starts with the same initial assignment as SD and we will show that it
does not leave $N_t$.

At some steps the distance between $\vu_d(t)$ and $\vu_l(t)$ remains the same,
and at some it changes.  Let $\vu_d,\vu_l$ be the assignments produced by the algorithms 
after $\tau$ changes have taken place, and $D_\tau$ be the distance between them. If $2 \gm_2 n^\al<D_\tau< \gm_1 n^\al$ we have $\Expect{D_{\tau+1}|D_\tau}< D_{\tau} - \frac{1}{3}$. Indeed, the number of variables voted to be zero does not exceed $\gm_2n^\alpha$ and is at least twice less than number of variables that differ in $\vu_d(t)$ and $\vu_l(t)$. Since any change in the distance between the assignments happens if and only if a variable voted to be 0 or a variable at which $\vu_d(t)$ and $\vu_l(t)$ are different, we have the required inequality. Now we can apply
Lemma~\ref{lem:strip_cross} for $D$ setting $L = 2\gm_2n^{\alpha},
H = 3\gm_2n^{\alpha}, d = 1/3$ and get that probability of LS leaving
$N_t$ is less than $\varrho n e^{-n^{\alpha}/6}$.
\end{proof}

\begin{corollary}\label{cor:they_are_close}
For $\vf \in \Plpsi$
there is a constant $\alpha_3$ such that distance between 
$\vu_d(t)$ and $\vu_l(t)$ defined in Lemma~\ref{lem:coupling1} whp does not exceed $n^{\alpha_3}$.
\end{corollary}

We say that a variable {\em plays $d$-righteously in a run of LS} if every time it is considered for flipping it is $d$-righteous.
Combining corollaries \ref{cor:run_is_safe} and \ref{cor:they_are_close} we 
obtain the following

\begin{lemma}\label{lem:innocence}
For any $d$ there is $\alpha_4<1$ such that, for a run of LS on $\vf\in\Plpsi$ whp the number of variables that do not play $d$-righteously is bounded above by $n^{\alpha_4}$. 
\end{lemma}

\begin{proof}
From  Corollaries \ref{cor:run_is_safe} and \ref{cor:they_are_close} it follows
that whp at every step of LS the number of variables
that are not $d$-righteous is less than $n^{\tilde{\alpha}}$, for some $\tilde\alpha$.

Therefore denoting the number of
different assignments considered by LS by $T$ (note that $T\le\vr n$) and observing that at each step the probability to consider a variable voted to be 0 is $n^{\tilde{\alpha}-1}$ we obtain the following upper bound for the expectation of
the number of non-$d$-righteous variables throughout the run:
$$
T n^{\tilde{\alpha} -1} \leq
 \kappa' n (\ln n) n^{\tilde{\alpha} - 1} =
 \kappa' n^{\tilde{\alpha}} \ln n \leq
 n^{\tilde{\alpha}+ \varepsilon}
$$
for arbitrary $\varepsilon$ with $\tilde{\alpha}+2\varepsilon<1$. We apply Markov inequality and obtain
$
\Prob{I > n^{\tilde{\alpha} + 2\varepsilon}} \leq n^{-\varepsilon},
$
where $I$ denotes the number of variables that do not play $d$-righteously.
Now $\alpha_4$ can be set to be $\tilde{\alpha} + 2\varepsilon$.
\end{proof}

\newcommand{\SCl}{\mathcal{C}}
\newcommand{\dom}{\mathop{\rm dom}}
\newcommand{\im}{\mathop{\rm im}}
\newcommand{\var}{\mathop{\rm var}}

A clause $(\ov x,\ov y,z)$ is called a \emph{cap support}
if there are $w_1,w_2$ such that $(x,w_1,w_2,y,z)$ is a cap in $\vf$.
For a formula $\psi$ we denote the set of variables that occur in it by $\var(\psi)$. For a set of clauses $K$ we  denote by $\bigwedge K$ a CNF formula constructed by conjunction of the clauses. For the sake of simplicity we will write $\var(K)$ instead of $\var\left(\bigwedge K\right)$. In what follows it will be convenient to view a CNF as a sequence of clauses. Note that representation of a CNF is quite natural when we sample a random CNF by generating random clauses. This way every clause occupies certain position in the formula. 
For a set of positions $P$ we denote the formula obtained from $\vf$ by removing all clauses except for occupying positions $P$ by $\vf \downarrow_P$. The set of variables occurring in the clauses in positions in $P$ will be denoted by $\var(P)$. 

We denote by $\mathcal{C}$ the set of all possible clauses over $n$ variables.
Let us fix a real constant $\nu<1$. We will need the following notation:
\begin{itemize}
\item let $[k]$ denote the set of the first $k$ positions of clauses in $\vf$, $V$ be the set of all variables in $\vf$;
\item let $S^{\vf,\nu}$ be the set of positions from
$[n^\nu]$ occupied by clauses that are cap supports in $\vf$, and 
$L^{\vf,\nu}$ the set of variables that occur in clauses in positions $S^{\vf,\nu}$;
\item let $T^{\vf,\nu}$ be set of positions of $\vf$ occupied by clauses containing a variable from $L^{\vf,\nu}$;
\item let $U^{\vf,\nu}$ be the set of positions in $\vf$ occupied by clauses containing a variable 
from $\var\left(\vf \mathop{\downarrow}_{[n^\nu] \setminus S^{\vf,\nu}}\right)$;
\item finally, let $R^{\vf,\nu}=[\varrho n] \setminus (S^{\vf,\nu} \cup U^{\vf,\nu})$;
\item let also $M^{\vf,\nu} = \var(T^{\vf,\nu})$ and $N^{\vf,\nu} = \var(U^{\vf,\nu})$. 
\end{itemize}

Fig.~\ref{fig:sets} pictures the notation just introduced.

\begin{figure}[h]
\centerline{\includegraphics[totalheight=2.5cm,keepaspectratio]{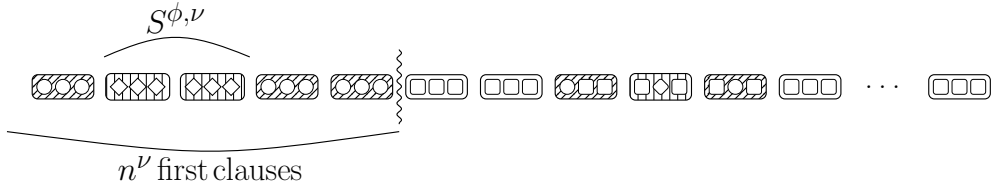}}
\caption{A scheme of a 3-CNF. Every clause is shown as a rectangle with its literals represented by squares inside the rectangle. 
Literals corresponding to variables from $L^{\phi \nu}$ and from $\var\left(\vf \mathop{\downarrow}_{[n^\nu] \setminus S^{\vf,\nu}}\right)$ are shown as    
diamonds and 
circles, respectively. Shaded rectangles with vertical and diagonal lines represent clauses  from $T^{\phi \nu}$ and $U^{\phi \nu}$, respectively.}
\label{fig:sets}
\end{figure}

\begin{lemma}\label{lem:mu_and_nu}
If $\rho \le \kappa \ln n$ and $\kappa<\frac{7}{6}$ then there is $\mu_0$ such that 
for any $\mu<\mu_0$
there is $\nu < 1$ such that whp:
\begin{itemize}
\item[(1)] 
$|S^{\vf,\nu}|\sim n^\mu$;
\item[(2)] 
$M^{\vf,\nu} \cap N^{\vf,\nu} = \varnothing$, that is variables from clauses from $U^{\vf,\nu}$ do not appear in the same clauses with variables from $S^{\vf,\nu}$;
\item[(3)] 
$|M^{\vf,\nu}| = 3|T^{\vf,\nu}|$, that is no variable occurs twice in the clauses from $T^{\vf,\nu}$.
\end{itemize} 
\end{lemma}

\begin{proof}
It follows from Lemma~\ref{lem:close-isolated} that for $\varrho \leq \kappa \ln n, \kappa< \frac 76$ there exists $\alpha, 0< \alpha< 1$ such that the number of caps in the formula is $\sim n^\alpha$.
We set 
$$
\mu_0 = \alpha/2,\qquad \nu = \mu + 1 - \alpha.
$$

(1) For a subset $R$ of all positions of clauses in $\phi$ let $\mathcal{C}_R$ denote event ``$R$ is exactly the set of
positions occupied by cap supports''. Obviously for any sets $R_1, R_2, |R_1| = |R_2|$ we have $\Prob{\mathcal{C}_{R_1}} = \Prob{\mathcal{C}_{R_2}}$. Thus positions of the cap supports are selected uniformly at random without repetition. By straightforward computation we have expectation of the number of cap supports among first $n^\nu$ clauses equal approximately $n^\alpha \cdot n^{\nu - 1} = n^{\mu + 1 - \alpha - 1 + \alpha} = n^\mu$ and variance is bounded above by the expectation, so it follows from Chebyshev inequality that random variable ``number of cap supports among first $n^\nu$ clauses'' is whp $\sim n^\mu$.

(2) By Lemma~\ref{lem:random}(2) whp there is no variable that occurs in more than $\ln^2n$ clauses. Therefore $|M^{\vf,\nu}|=O(n^\mu\ln^2n)$ and $|N^{\vf,\nu}|=O(n^\nu\ln^2n)$. These sets are randomly chosen from an $n$-element set, and therefore the probability they have a common element is at most $n^{\mu+\nu-1}\ln^4n$. Due to definition of $\mu$ and $\nu$ we have $\mu+\nu-1 < {\alpha / 2 + \alpha/2 + 1 - \alpha - 1} = 0$. 

(3) Since whp $|T^{\vf,\nu}|=O(n^\mu\ln^2n)$, the probability that two clauses from this set share a variable is bounded above by $n^{2\mu - 1}\ln^4n$. We have $2\mu - 1 < \alpha - 1 <0$ so this probability tends to~0. 
\end{proof}

Let us fix a formula $\vf$ selected accordingly $\Pl$ and $\mu<\frac15$, and let $\nu$ correspond to $\mu$ as in Lemma~\ref{lem:mu_and_nu}. 
Let $T_0$ and $U_0$ be subsets of $[\varrho n]$ such that $T_0 \cap U_0 = \varnothing$, $[n^\nu] \subseteq T_0 \cup U_0$ and let $S_0 = T_0 \cap [n^\nu]$.
We denote by $H_{T_0 U_0}$ a hypothesis stating that $\vf$ is such that $S^{\vf,\nu} = S_0$, $T^{\vf,\nu} = T_0$,  $U^{\vf,\nu} = U_0$ and also $M^{\vf,\nu} \cap N^{\vf,\nu} = \varnothing$, $\left|M^{\vf,\nu}\right| = 3 \left|T^{\vf,\nu} \right| $.

\begin{lemma}\label{lem:hypothesis} 
If for an event $E$ there is a sequence $\delta(n) \goes_{n\goes \infty} 0$ such that for all pairs $(T_0, U_0)$, $|T_0 \cup U_0| < n^{2 \nu}$ we have $\Prob{E | H_{T_0 U_0}} \leq \delta(n)$ then $\Prob{E} \goes_{n \goes \infty} 0$. 
\end{lemma}

\begin{proof}
We can bound probability of event $E$ as
\begin{eqnarray*}
\Prob{E}&\leq& \sum_{{\tiny T_0, U_0 : |T_0 \cup U_0| < n^{2 \nu} }}
(\Prob{E | H_{T_0 U_0}} \Prob{H_{T_0 U_0}} \\
& & \qquad + \Prob{M^{\vf,\nu}\cap N^{\vf,\nu} \neq \varnothing \text{ or } \left|M^{\vf,\nu}\right| < 3 \left|T^{\vf,\nu} \right| \text{ or } |T_0 \cup U_0| \geq n^{2 \nu} }) \\
 &\leq& \delta(n) + \Prob{M^{\vf,\nu}\cap N^{\vf,\nu} \neq \varnothing} + \Prob{\left|M^{\vf,\nu}\right| < 3 \left|T^{\vf,\nu} \right|} + \Prob{|T_0 \cup U_0| \geq n^{2 \nu}}.
\end{eqnarray*}
By Lemma~\ref{lem:mu_and_nu} probabilities of events $M^{\vf,\nu}\cap N^{\vf,\nu} \neq 
\varnothing$ and $\left|M^{\vf,\nu} \right| < 3 \left|T^{\vf,\nu} \right|$ tend to 0 as $n$ approaches infinity. By Lemma~\ref{lem:random} (2) we have $|T_0 \cup U_0| < n^{2 \nu}$ whp. Thus we obtain the result.
\end{proof}

\begin{observ}\label{its_random}
If $\vf$ is selected according to $\Pl$ conditioned to $H_{T_0 U_0}$ 
then formula $$\vf \downarrow _{[\varrho n] \setminus (T_0 \cup U_0)}$$ 
has the same distribution as if it was generated by picking clauses 
from all clauses over variables\lb $V \setminus \var([n^\nu])$ uniformly at random.
\end{observ}

\begin{proof}
Let $\mathcal{C}'$ be the set of all clauses over variables in $V \setminus \var([n^\nu])$ and $R_0=[\vr n]\setminus(T_0\cup U_0)$. Take a formula $\psi$ such that positions from $R_0$ of this formula are occupied by clauses from $\mathcal{C}'$. It suffices to observe that the number of formulas $\psi'$ such that $\psi'\downarrow_{R_0}=\psi\downarrow_{R_0}$, $S^{\psi',\nu}=S_0$, $T^{\psi',\nu}=T_0$, $U^{\psi',\nu}=U_0$ is the same for any $\psi$. 
So since all possible formulas over variables from some set are equiprobable a random formula is generated by random sampling of clauses.
\end{proof}

\begin{proof}[of Proposition~\ref{pro:failing}]
We will bound probability of success of Local Search under a hypothesis of the form ${H_{T_0 U_0}}$ and apply Lemma~\ref{lem:hypothesis} to get the result. Let $\al_4$ be the exponent corresponding to $\vr$ by Lemma~\ref{lem:innocence}, and choose $\mu$ and $\nu$ such that $\al_4+2\mu<1$.

Let $M=M^{\vf,\nu}$ and $L=L^{\vf,\nu}$. We split formula $\vf$ into $\vf_1 = \vf\downarrow_{T_0}$ and $\vf_2 = \vf\downarrow_{[\varrho n] \setminus T_0}$ and first consider a run of LS applied to $\vf_2$ only.
Formula $\vf_2$ can in turn be considered as the conjunction of $\vf_{2 1} = \vf\downarrow_{U_0}$ and $\vf_{2 2} = \vf \downarrow_{[\varrho n] \setminus (T_0 \cup U_0)}$. In Fig.~\ref{fig:sets} formula $\vf_{1}$ consists of clauses shaded with vertical lines, formula $\vf_{2 1}$ of clauses shaded with diagonal lines and formula $\vf_{2 2}$ of clauses that are not shaded. 
By Observation~\ref{its_random} formula $\vf_{2 2}$ is sampled according to 
$$
\Phi^{\tt plant}(n - \delta_1(n), n \varrho - \delta_2(n))
$$ 
modulo names of variables where $\delta_1(n)$ and $\delta_2(n)$ are $o(n)$. So formula $\vf_{2}$ is sampled according to 
$$
\Phi^{\tt plant}_{2 \mu}(n - \delta_1(n), n \varrho - \delta_2(n)).
$$
By Lemma~\ref{lem:innocence} the number of variables that do not play $2$-righteously during run of LS on $\vf_2$ is bounded from above by $n^{\alpha_4}$ for a certain $\al_3<1$.  

We consider coupling $(LS_\vf, LS_{\vf_2})$ of runs of LS on $\vf$ and $\vf_2$, denoting 
assignments obtained by the runs of the algorithm at step $t$ by $\vu_\vf(t)$ and 
$\vu_{\vf_2}(t)$ respectively. 
Let $K$ be the set of those variables which do not belong to $L$ (squares and circles in Fig.~\ref{fig:sets}). Formula $\vf_2$ is a 3-CNF containing only variables from $K$. For an assignment of values of all variables $\vu$ we will denote by $\vu|_K$ its restriction onto variables from $K$.
We make process $LS_\vf$ start with a random assignment $\vu_\vf(0) = \vu_\vf^0$ to all variables, and $LS_{\vf_2}$ with a random assignment $\vu_{\vf_2}(0) = \vu_{\vf_2}^0$ to variables in $K$, such that $\vu_\vf^0|_K = \vu_{\vf_2}^0$. 
Now the algorithms work as follows. At every step a random variable $x_i$ is chosen. Process $LS_\vf$ makes its step, and process $LS_{\vf_2}$ makes its step if $x_i \in K$.

Whp $LS_{\vf_2}$ will run with at most $n^{\alpha_4}$ variables that do not play $2$-righteously. Let $W$ denote the set of such variables. Variables in formula $\vf_1$ are selected uniformly at random so if $\alpha_4 + 2\mu < 1$ then whp set $M$
does not intersect with $W$. Hence, every time $LS_\vf$ considers some variable from $M$ it is $2$-righteous in $\vf_2$ and belongs to at most one clause of $\vf_1$. Therefore such a variable is at least $1$-righteous $\vf$ and is flipped to 1, or stays 1, whichever is to happen for $LS_{\vf_2}$.  Thus whp at every step of $(LS_\vf, LS_{\vf_2})$ we have $\vu_\vf(t)|_K = \vu_{\vf_2}(t)$. In the rest of the proof we consider only this highly probable case.
  
Consider some cap support $c_i = (\ov x_1, \ov x_4, x_5)$ occupying a position $i\in[n^\nu]$ and such that $x_1 = 0, x_4 = 1, x_5 = 0$ at time 0, and a set $P_{c_i}$ of variables occurring in clauses that contain variables $\var(c_i)$ (obviously $\var(c_i) \subseteq P_{c_i}$). Let $c_j$ be the clause that forms a cap with $c_i$. 
We say that a variable is {\em discovered} at step $t$ if it is considered for the first time at step $t$. Let $p_1, \dots, p_k$ be an ordering of elements of $P_{c_i}$ according to the step of their discovery. In other words if variable $p_1$ is the first variable from $P_{c_i}$ that is discovered, $p_k$ was the last. In the case some variables are not considered at all, we place them in the end of the list in a random order. Observe that all variables that play at least $1$-righteously are discovered at some step. All orderings of variables are equiprobable, hence, the probability of variables $\var(c_i)$ to occupy places $p_{k-2}, p_{k-1}$ and $p_k$ equals $3! / k (k-1) (k-2)$. We will call this ordering {\em unlucky}.

Let us consider what happens if the order of discovery of $P_{c_i}$ is unlucky. All variables in $P_{c_i} \setminus \var(c_i)$ play $1$-righteously, therefore once they are discovered by $LS_\vf$ they equal to 1.  Thus when $x_1, x_4, x_5$ are finally considered all clauses they occur in are satisfied, except for $c_j$. So variables $x_1, x_4, x_5$ do not change their values and the clause $c_j$ remains unsatisfied by the end of the work of $LS_\vf$. 

By Lemma~\ref{lem:random}(2) whp no vertex has degree greater than $\ln^2 n$, so the size of the set $P_{c_i}$ is bounded above by $3 \ln^2 n$. Thus the probability of event $Unluck(i) = $``order of discovery of $\var(c_i)$ is unlucky'' is greater than $\frac{1}{\ln^6 n}$. Thus, the expectation of $|\{ i | Unluck(i)\}|$
equals 
$$
\frac{|S_0|}{\ln^6 n} = \frac{n^\mu}{\ln^6 n}.
$$
Any variable whp occurs in clauses from $T^{\vf,\nu}$ at most once, hence there is no variable that occurs in the same clause with a variable from $c_{i_1}$ and a variable from $c_{i_2}$ for $i_1,i_2\in S_0$, $i_1\ne i_2$. This implies that events of the form $Unluck(i)$
are independent. Therefore random variable $|\{ i | Unluck(i)\}|$ is Bernoulli and, as its expectation tends to infinity, the probability that it equals to $0$ goes to 0. Since unlucky ordering of at least one cap support leads to failure of the LS this proves the result.
\end{proof}
  
\bibliographystyle{plain}

\end{document}